\documentclass[3p]{elsarticle}

\makeatletter
\def\ps@pprintTitle{%
 \let\@oddhead\@empty
 \let\@evenhead\@empty
 \def\@oddfoot{\centerline{\thepage}}%
 \let\@evenfoot\@oddfoot}
\makeatother

\usepackage{amssymb}
\usepackage{amsmath}
\usepackage{amsthm}
\usepackage{hyperref}

\theoremstyle{plain}
\newtheorem{thm}{Theorem}[section]
\newtheorem{lem}[thm]{Lemma}
\newtheorem{prop}[thm]{Proposition}
\newtheorem{cor}[thm]{Corollary}

\theoremstyle{definition}
\newtheorem{defn}{Definition}[section]
\newtheorem{exmp}{Example}[section]

\theoremstyle{remark}
\newtheorem*{rem}{Remark}

\journal{Information Processing Letters}

\begin{document}

\begin{frontmatter}

\title{On the longest common subsequence of Thue-Morse words}

\author{Joakim Blikstad}

\ead{joblikst@uwaterloo.ca}
\address{University of Waterloo, Canada}

\begin{abstract}
  The length $a(n)$ of the longest common subsequence of the $n$'th Thue-Morse
  word and its bitwise complement is studied.  An open problem suggested by
  Jean Berstel in 2006 is to find a formula for $a(n)$.  In this paper we prove
  new lower bounds on $a(n)$ by explicitly constructing a common subsequence
  between the Thue-Morse words and their bitwise complement.  We obtain the
  lower bound $a(n) = 2^{n}(1-o(1))$, saying that when $n$ grows large, the
  fraction of omitted symbols in the longest common subsequence of the $n$'th
  Thue-Morse word and its bitwise complement goes to $0$.  We further
  generalize to any prefix of the Thue-Morse sequence, where we prove similar
  lower bounds.
\end{abstract}


\end{frontmatter}

\section{Introduction}

The Thue-Morse sequence is a well known sequence in mathematics and computer
science, with many interesting properties.

The Thue-Morse sequence has a lot of self-symmetry in it, but is at the same
time cube-free and overlap-free (for a more in depth introduction to the
Thue-Morse sequence, see, for instance, Allouche and Shallit \citep{allouche}).

In 2006, \citet{berstel} formulated the problem of finding the length $a(n)$ of
the longest common subsequence between the $n$'th Thue-Morse word and its
bitwise complement.  By bitwise complement we mean replacing $0$ with $1$ and
$1$ with $0$.  This paper primarily studies $a(n)$ (sequence
\href{https://oeis.org/A297618}{A297618} on the \emph{Online Encyclopedia of
Integer Sequences} \cite{oeis}).  Since the Thue-Morse words are prefixes of
length $2^k$ for some $k$, of the Thue-Morse sequence, a natural generalization
is to consider other length prefixes of the Thue-Morse sequence.  This paper
also studies $b(n)$, the longest common subsequence between the length $n$
prefix of the Thue-Morse sequence and its bitwise complement (sequence
\href{https://oeis.org/A320847}{A320847}).

\begin{exmp}
  The first few values of $a(n)$ and $b(n)$ are:
  \begin{align*}
    a(1) &= 1 & b(1) &= 0 \\
    a(2) &= 2 & b(2) &= 1 \\
    a(3) &= 5 & b(3) &= 1 \\
    a(4) &= 12 & b(4) &= 2 \\
    a(5) &= 26 & b(5) &= 3 \\
    a(6) &= 54 & b(6) &= 4 \\
  \end{align*}
\end{exmp}

To show a lower bound for $a(n)$, it suffices to
construct a common subsequence of
the Thue-Morse words and their bitwise complements.
This is what is done in this paper, using the symmetries
of the sequence.
In particular,
we provide a recursive construction for such a common subsequence,
which has length at least
$2^{n}(1-\mathcal{O}(n^{-\log_2 3})) = 2^{n}(1-o(1))$.

This new lower bound is interesting as it means that $\frac{a(n)}{2^{n}}$
goes to $1$, that is when $n$ grows large the longest common subsequence
will only omit a vanishingly small fraction of symbols.

\section{Setup}
There are many equivalent
definitions of the Thue-Morse sequence and Thue-Morse words.
We will define them using morphisms.
\begin{defn}
  A morphism over an alphabet $\Sigma$ is a function
  $m:\Sigma^*\to\Sigma^*$ that satisfies $m(xy)=m(x)m(y)$
  (concatenation)
  for all $x,y\in \Sigma^*$. Note that this means $m$ is uniquely defined
  by its behaviour on $\Sigma$.
\end{defn}
\begin{defn}
  Let $\mu$ denote the morphism on $\{0,1\}$ defined by $\mu(0) = 01$ and $\mu(1) = 10$.
\end{defn}

There are some basic properties that follow directly from the definition.
\begin{prop}
  If $n\ge 0$ then
  \begin{enumerate}[(i)]
    \item $\mu^{n}(1) = \overline{\mu^{n}(0)}$ where $\overline{z}$ denotes taking the bitwise complement of $z$ (i.e., swapping 0s and 1s).
    \item $\mu^{m + n}(0) = \mu^{m}(\mu^{n}(0))$.
    \item $\left|\mu^{n}(0)\right| = 2^n$.
    \item $\mu^{n+1}(0) = \mu^{n}(0) \mu^n(1)$ and $\mu^{n+1}(1) = \mu^{n}(1) \mu^n(0)$.
  \end{enumerate}
  \label{prop:basic}
\end{prop}
\begin{proof}
  (i) follows from the symmetry (between $0$ and $1$) in the definition of $\mu$.
  (ii) holds for all morphisms. (iii) follows from an induction argument
  since $|\mu(x)| = 2|x|$ for every binary string $x$.
  (iv) can be seen from $\mu^{n+1}(0) = \mu^{n}(\mu(0)) = \mu^{n}(01)
  = \mu^{n}(0) \mu^n(1)$.
\end{proof}
\begin{defn}
  We call $\mu^n(0)$ the $n$'th \emph{Thue-Morse word}.
  We also say the \emph{Thue-Morse sequence}, denoted by $\mathbf{t}$,
  is the the unique fixed point of $\mu$
  (extended to the domain of infinite binary strings) beginning with a~$0$.
  See \citet{allouche} for why such a fixed point exists and is unique.
\end{defn}
\begin{defn}
  Denote by $a(n)$ the length of the longest common subsequence of
  $\mu^n(0)$ and $\mu^n(1)$.
  Similarly, denote by $b(n)$ the length of the longest common subsequence of
  the prefix of length $n$ of the Thue-Morse sequence and its bitwise complement.
\end{defn}
\begin{exmp} 
The first few Thue-Morse words are
  \[\mu^{0}(0) = 0,\quad
  \mu^{1}(0) = 01,\quad
  \mu^{2}(0) = 0110, \quad
  \mu^{3}(0) = 01101001.
  \]
The Thue-Morse sequence starts as follows $\mathbf{t} = 0110100110010110\ldots$
\end{exmp}
\begin{rem}
  The Thue-Morse words are sometimes defined by
  the recurrence relation in Proposition~\ref{prop:basic} part~(iv),
  and then the Thue-Morse sequence as the infinite application of this
  rule. We see that $n$'th Thue-Morse word is the prefix of length $2^n$
  of the Thue-Morse sequence. This also means that $b(2^n) = a(n)$.
\end{rem}

We also need the following proposition,
for which the proof can be found in \citep{allouche}.
\begin{prop}
  If $\mathbf{t} = t_0 t_1 t_2\ldots$ are the symbols of the Thue-Morse sequence
  we have $t_{2n} = t_{n}$ and $t_{2n+1} = \overline{t_{n}}$ for all
  $n\ge 0$. Moreover, $t_{n}$ equals the parity of the number of ``1'' bits in
  the binary representation of $n$.
  \label{prop:thue}
\end{prop}
\begin{cor}
  The $(2i)$'th digit of $\mu^n(0)$ is the same as the
  $(2i+1)$'th digit of $\mu^n(1)$ (where we use zero-indexing).
  \label{cor:thue}
\end{cor}
\begin{proof}
  The $(2i)$'th digit of $\mu^n(0)$ is $t_{2i} = t_i$,
  and the $(2i+1)$'th digit of $\mu^n(1)$ is $\overline{t_{2i+1}}
  =t_{i}$, by the above proposition.
\end{proof}

\section{Construction of a common subsequence} \label{sec:construct}
We are now ready for a construction of a common subsequence between $\mu^n(0)$
and $\mu^n(1)$ when $n = 2^k$ is a power of $2$.  We call this common
subsequence $CS(k)$, and define it recursively.

\begin{itemize}
  \item
    When $k = 0, n = 2^{0} = 1$, and we define $CS(0) = 0$, a subsequence of
    $\mu(0) = \underline{0}1$ and $\mu(1) = 1\underline{0}$.

  \item
    For $k \ge 1$, $CS(k)$ will be defined recursively as follows.

    Let $n = 2^k$ and $m = 2^{k-1}$. Say $X = \mu^{n}(0)$ and $Y = \mu^{n}(1)$,
    that is, we are constructing $CS(k)$ as a common subsequence of $X$ and
    $Y$.  Write $X$ and $Y$ as concatenations of $2^m$ blocks of size $2^m$
    (since $|X| = |Y| = 2^n = (2^m)^2$ this is possible), say
    \begin{align*}
      X &= x_0 x_1 \cdots x_{2^m-1} \\
      Y &= y_0 y_1 \cdots y_{2^m-1}.
    \end{align*}
    Since $X = \mu^{2^k}(0) = \mu^{2^{k-1}}(\mu^{2^{k-1}}(0))$, each $x_i$ is
    one of $\mu^{m}(0)$ or $\mu^{m}(1)$. Similarly each $y_i$ is one of
    $\mu^{m}(0)$ or $\mu^{m}(1)$.  It is also worth noting that $x_i =
    \mu^m(d)$ if the $i$'th digit of $\mu^m(0)$ is $d$, and similarly $y_i =
    \mu^m(d)$ if the $i$'th digit of $\mu^m(1)$ is $d$.

    Now we compare $x_i$ to $y_{i+1}$ for $0\le i < 2^m-1$, and find a common
    subsequence $cs_i$ between them.
    \begin{itemize}
      \item
        When $i$ is even, $x_i = y_{i+1}$ by Corollary~\ref{cor:thue},
        so we take $cs_i = x_i$.
      \item
        When $i$ is odd, either $x_i$ and $y_{i+1}$ are the same, or one is
        $\mu^{m}(0)$ and the other is $\mu^{m}(1)$.  If they are the same we
        take $cs_i = x_i$, otherwise $cs_i = CS(k-1)$.
    \end{itemize}
    We then let $CS(k)$ be the concatenation of the $cs_i$'s.  \end{itemize}

\begin{exmp}
  The common subsequence $CS(0), CS(1)$, and $CS(2)$
  are underlined below:
\begin{align*}
  CS(0): \quad \mu^1(0)=&\ \underline{0}1 \\
               \mu^1(1)=&\ 1\underline{0} \\
  CS(1): \quad \mu^2(0)=&\ \underline{01}\ 10\\
               \mu^2(1)=&\ 10\ \underline{01}\\
  CS(2): \quad \mu^4(0)=&\ \underline{0110}\ 10\underline{01}\ \underline{1001}\ 0110 \\
               \mu^4(1)=&\ 1001\ \underline{0110}\ \underline{01}10\ \underline{1001}
\end{align*}
\end{exmp}
\begin{rem}
$CS(k)$ is not necessarily the longest common subsequence.
For example
\begin{align*}
  \mu^4(0)=&\ \underline{011}0\ \underline{1001}\ \underline{100}1\ \underline{01}10 \\
  \mu^4(1)=&\ 1\underline{0}0\underline{1}\ 0\underline{110}\ \underline{0110}\ 1\underline{0}0\underline{1}
\end{align*}
is the longest common subsequence between $\mu^4(0)$ and $\mu^4(1)$, which has length~$12$, while $|CS(2)| = 10$.
\end{rem}

\section{Analysis of length} \label{sec:analysis}
In this section we analyse the length of the common subsequence $CS(k)$
constructed in the previous section.
\begin{defn}
  For an integer $k\ge 0$, let $f(k) = |\mu^{2^{k}}(0)| - |CS(k)|
  = 2^{2^{k}} - |CS(k)|$
be the number of symbols omitted by the common subsequence $CS(k)$.
\end{defn}
\begin{rem}
$f(0) = 1$, as $|CS(0)| = 1$.
\end{rem}

When constructing $CS(k+1)$, all the even indexed blocks
(of size $2^{2^k}$)
in $\mu^{2^{k+1}}(0)$
are chosen to be in $CS(k+1)$. So only the
odd indexed blocks can contribute to $f(k+1)$.
The last block will be completely omitted,
and for the other blocks in odd positions
we either miss $f(k)$ if matching
$\mu^{2^{k}}(0)$ with $\mu^{2^{k}}(1)$ recursively,
or miss nothing if choosing to include the complete block.
This leads us to the following lemma.
\begin{lem}
  For every integer $k\ge 0$
  \[f(k+1) \le 2^{2^k} + \left(2^{2^k-1}-1\right)f(k).\]
  \label{lem:bound1}
\end{lem}
\begin{proof}
  The last block has size $2^{2^k}$,
  and there are $(2^{2^k-1}-1)$
  other odd indexed blocks, and in each we miss at most $f(k)$.
  So the lemma follows from the above discussion.
\end{proof}

We are now ready to prove an upper bound on $f(k)$.
\begin{lem}
  For every integer $k\ge 0$, $f(k)\le 2^{2^k-k+1}-2$.
  \label{lem:bound2}
\end{lem}
\begin{proof}
  We proceed by induction on $k$.

The inequality clearly holds for $k = 0$ since $f(0) = 1 \le 4-2 = 2^{2^{0}-0+1}-2$

Now suppose the inductive assertion holds for $k = s \ge 0$, that is $f(s)\le 2^{2^s-s+1}-2$.
Using Lemma~$\ref{lem:bound1}$ and the induction hypothesis we have
\begin{align*}
  f(s+1) &\le 2^{2^s} + (2^{2^{s}-1}-1)f(s) \\
  &\le 2^{2^s} + (2^{2^s-1}-1)(2^{2^s-s+1}-2) \\
  &= 2^{2^s} + 2^{2^s-1+2^s-s+1} - 2^{2^{s}-1}\cdot 2 - 2^{2^s-s+1} + 2 \\
  &= 2^{2^{s+1}-(s+1)+1} - 2^{2^s-s+1} + 2.
\end{align*}
  Note that $2^{2^s-s+1} \ge 4$ for all integers $s\ge 0$,
  since $2^s-s\ge 1$ for all integers $s\ge 0$. Thus
\begin{align*}
  f(s+1) &\le 2^{2^{s+1}-(s+1)+1} - 2^{2^s-s+1} + 2\\
  &\le 2^{2^{s+1}-(s+1)+1} - 2.
\end{align*}
This concludes the induction proof.
\end{proof}

By Lemma~$\ref{lem:bound1}$ it follows that
$f(k) \le 2^{2^k-k+1}-2 \le 2^{2^k-(k-1)}$ for all $k\ge 0$.
This means that the length of our constructed common subsequence
$CS(k)$ of $\mu^{n}(0)$ and $\mu^{n}(1)$
where $n = 2^k$ must be at least $2^{n} - f(k) \ge 2^{2^{k}}-2^{2^{k}-(k-1)} = 2^{2^k}(1-2^{-(k-1)}) =  2^{n}(1-\frac{1}{n/2})$. This proves
the following theorem.
\begin{thm}
For $k\ge 0$ and $n = 2^k$:
  \[|CS(k)| \ge 2^{n}\left(1-\frac{1}{n/2}\right)
  = 2^{2^{k}}\left(1-\frac{1}{2^{k-1}}\right).\]
\label{thm:pow2}
\end{thm}

\section{Extension to all $n$}
Up to this point we have only considered the common subsequence of
$\mu^n(0)$ and $\mu^n(1)$ where $n = 2^k$ for some $k\ge 0$.
We wish to extend our construction to work for arbitrary $n$.

If $n\ge 1$ and $n\not= 2^k$, then say $2^k < n < 2^{k+1}$ for some integer $k\ge 0$.
Write
\begin{align*}
  \mu^n(0) &= \mu^{n-2^k}(\mu^{2^k}(0)) \\
  \mu^n(1) &= \mu^{n-2^k}(\mu^{2^k}(1)).
\end{align*}
This is saying that $\mu^n(x)$ ($x \in \{0,1\}$) can be written
as $2^{n-2^k}$ blocks, where each block is either $\mu^{2^k}(0)$
or $\mu^{2^k}(1)$.
We can concatenate $2^{n-2^k}$ copies of the subsequence $CS(k)$
to obtain a common subsequence of $\mu^n(0)$ and $\mu^n(1)$,
i.e., we use our previous construction for each of the blocks independently.
Using Theorem~$\ref{thm:pow2}$ we see that the length of this common subsequence is at least $2^{n-2^k} (2^{2^k}(1-\frac{1}{2^{k-1}})) \ge 2^n (1-\frac{1}{n/4})$,
since $\frac{n}{4} \le 2^{k-1}$ by choice of $k$.
We thus get a similar result as Theorem~$\ref{thm:pow2}$ for arbitrary $n$.
\begin{thm}
  For every $n\ge 1$, there exists a common subsequence between
  $\mu^n(0)$ and $\mu^n(1)$ with length at least
  \[2^n\left(1-\frac{1}{n/4}\right).\]
  \label{thm:word-bound}
\end{thm}
\begin{cor}
  $a(n) = 2^{n}(1-\mathcal{O}(n^{-1}))$, or more generally
  $a(n) = 2^{n}(1-o(1))$.
  \label{cor:an}
\end{cor}

We can generalize the result further to all prefixes of the Thue-Morse sequence.
Let $\mathbf{t}_n$ be the prefix of length $n$ of the Thue-Morse sequence,
and $\overline{\mathbf{t}}_n$ its bitwise complement.
Based on the binary representation of the number $n$,
$\mathbf{t}_n$ and $\overline{\mathbf{t}}_n$ can be split up into
at most $\lfloor\log_2(n)\rfloor+1$
blocks, each with a size which is a power of $2$.
We will assume the blocks are in order of decreasing size, so that
a block of size $2^k$ is either $\mu^k(0)$ or $\mu^k(1)$.
Then common subsequences satisfying the inequality in Theorem~\ref{thm:word-bound}
for these blocks can be concatenated
to form a common subsequence between
$\mathbf{t}_n$ and $\overline{\mathbf{t}}_n$.
To bound the length of this common subsequence we use the following lemma:

\begin{lem}
  $\sum_{k=1}^{s} \frac{2^{k}}{k} \le \frac{2^{s+2}}{s}-1$ for all $s\ge 1$.
  \label{lem:sum}
\end{lem}
\begin{proof}
  We prove the inequality by induction on $s$.

  For $s = 1$ we have $\sum_{k=1}^{s} \frac{2^{k}}{k} = 2 \le 7 = \frac{2^{s+2}}{s}-1$,
  and for $s = 2$ we have $\sum_{k=1}^{s} \frac{2^{k}}{k} = 4 \le 7 = \frac{2^{s+2}}{s}-1$.

  Now suppose $s\ge 2$ and
  $\sum_{k=1}^{s} \frac{2^{k}}{k} \le \frac{2^{s+2}}{s}$.
  This means that
  \begin{align*}
  \sum_{k=1}^{s+1} \frac{2^{k}}{k} 
  = \sum_{k=1}^{s} \frac{2^{k}}{k} + \frac{2^{s+1}}{s+1}
  \le \frac{2^{s+2}}{s}-1 + \frac{2^{s+1}}{s+1}
  = \frac{2^{s+1}(3s+2)}{s(s+1)} -1
  \le \frac{2^{s+1}(4s)}{s(s+1)} -1
  = \frac{2^{s+3}}{(s+1)} -1,
  \end{align*}
  which concludes the induction proof.
\end{proof}

Now we continue to analyse the common subsequence 
between $\mathbf{t}_n$ and $\overline{\mathbf{t}}_n$.
This subsequence
omits at most $\frac{2^{k+2}}{n}$ symbols
for the block of size $2^{k}$ (by Theorem~\ref{thm:word-bound}).
There is at most one block of size $2^k$ for each
$1\le k \le \lfloor \log_2(n)\rfloor$.
The potential block of size $1 = 2^{0}$ will miss at most one symbol.
Hence at most
\[1+\sum_{k=1}^{\lfloor\log_2(n)\rfloor} \frac{2^{k+2}}{k} =
1+4 \sum_{k=1}^{\lfloor\log_2(n)\rfloor} \frac{2^{k}}{k}\]
symbols are omitted, which by Lemma~\ref{lem:sum} is at most
\[1+4\left(\frac{2^{\lfloor\log_2(n)\rfloor+2}}{\lfloor \log_2(n)\rfloor}-1\right)
= \frac{2^{\lfloor\log_2(n)\rfloor+4}}{\lfloor \log_2(n)\rfloor}-3
\le \frac{n}{\lfloor\log_2(n)\rfloor/16}.\]
This proves the following theorem.

\begin{thm}
  For all $n\ge 1$, there exists a common subsequence between
  $\mathbf{t}_n$ and $\overline{\mathbf{t}}_n$ with length at least
  \[n\left(1-\frac{1}{\lfloor\log_2(n)\rfloor/16}\right).\]
\end{thm}
\begin{cor}
  $b(n) = n(1-\mathcal{O}(\frac{1}{\log n}))$, or more generally
  $b(n) = n(1-o(1))$.
  \label{cor:bn}
\end{cor}

\section{Strengthening the analysis}
The constructed common subsequence $CS(k)$,
and the generalizations in the previous section,
does in fact have a slightly better asymptotic behaviour
than what was proven in Section \ref{sec:analysis}.

The previous length analysis was based on Lemma~\ref{lem:bound1}
which states that
$f(k+1)\le 2^{2^{k}} + \left(2^{2^{k}-1} - 1\right)f(k)$.
This inequality is only tight when
all $x_i \neq y_{i+1}$ for odd $0\le i < 2^{m}-1$,
using the same notation as in Section~\ref{sec:construct}.
However, we can get a better bound on $f(k+1)$ in terms
of $f(k)$ by estimating how many
of the blocks $x_i$ and $y_{i+1}$ are equal for odd $i$.

\begin{lem}
  If $\mathbf{t} = t_0 t_1 t_2\ldots$ are the digits of the Thue-Morse sequence,
  then $t_n = t_{n+1}$ if and only if $n$ written
  in binary ends with a block of $1$'s with odd length.
\end{lem}
\begin{proof}
  We use Proposition~\ref{prop:thue}.
  $t_n = t_{n+1}$ if an only if $n$ and $n+1$
  have the same number of ``$1$'' bits modulo 2, when written in binary.
  This condition is equivalent to $n$
  ending with a block of $1$'s of odd length when written in binary.
\end{proof}
\begin{lem}
  Let $eq(n) = |\{i : 0\le i < 2^n-1\text{ and } t_i = t_{i+1}\}|$.
  Then 
  \begin{align*}
    eq(n) = \begin{cases}
      \frac{1}{3}(2^n-1) & \text{if $n$ is even}\\
      \frac{1}{3}(2^n-2) & \text{if $n$ is odd}\\
    \end{cases}.
  \end{align*}
  \label{lem:sn}
\end{lem}
\begin{proof}
  For a fixed $n$, we count how many $n$-bit numbers (except $2^n-1$) which
  ends with a block of $1$'s of odd length.  We can fix the $n$-bit number to
  end with a ``$0$'' followed by $2k-1$ ``$1$''s, for different values of $k$,
  and then have $2^{n-2k}$ possibilities for the leading digits.
  This works as we do not wish to count $2^n-1$, which is the unique $n$-bit binary
  number with all ``1''s.
  \begin{itemize}
    \item 
      If $n = 2m$ is even $eq(n) = \sum_{k=1}^{m} 2^{n-2k} = \frac{1}{3}(2^n-1)$.
    \item 
      If $n=2m+1$ is odd, then $eq(n) = \sum_{k=1}^{m} 2^{n-2k} = \frac{1}{3}(2^n-2)$.
      \qedhere{}
  \end{itemize}
\end{proof}

By Proposition~\ref{prop:thue}
we see that 
\[x_{2i+1} = y_{2i+2} \iff
t_{2i+1} = \overline{t_{2i+2}} \iff
\overline{t_i} = \overline{t_{i+1}} \iff
t_i = t_{i+1}\]

By Lemma~\ref{lem:sn} we thus know that when constructing
$CS(k+1)$, exactly
$eq(2^{k}-1)$
of the odd indexed blocks will already be equal. Hence
exactly $(2^{2^{k}-1}-1)-eq(2^{k}-1)$
of the $(x_i, y_{i+1})$ pairs will need to be recursively matched
using $CS(k)$. This leads to the following improved version of Lemma~\ref{lem:bound1}:
\begin{lem}
  For every integer $k\ge 1$,
  \[f(k+1) = 2^{2^k} + \left(2^{2^k-1}-1 - eq(2^{k}-1)\right)f(k)
  = 2^{2^k} + \left(2^{2^k-1} -1- \frac{1}{3}(2^{2^k-1}-2)\right)f(k).\]
\end{lem}
\begin{rem}
  From the above lemma, we can solve for $f(k)$ exactly. The
  first few values for $k\ge 0$ are:
  \[f(k) = 1, 2, 6, 46, 4166, 91071806, 130383480383828886, \ldots\]
\end{rem}
\begin{cor}
  Let $w = \log_2(3)\approx 1.58$.
  For every integer $k\ge 1$, $f(k+1) \le 2^{2^k} + 2^{2^k-w}f(k)$.
  \label{cor:bound1}
\end{cor}
\begin{proof} If $k\ge 1$, we have by the lemma
  \[f(k+1) = 2^{2^k} + \left(2^{2^k-1} -1- \frac{1}{3}(2^{2^k-1}-2)\right)f(k)
  \le
  2^{2^k} + \frac{2}{3}2^{2^k-1}f(k) = 
  2^{2^k} + 2^{2^k-w}f(k).\]
\end{proof}
By a similar induction proof as in Lemma \ref{lem:bound2}
we get a new upper bound on $f$.
\begin{thm}
  Let $w = \log_2(3)\approx 1.58$.
  For every integer $k\ge 0$, $f(k)\le 2^{2^k-wk+3}-6$.
  \label{thm:bound2}
\end{thm}
\begin{proof}
  We proceed by induction on $k$.

  It is easy to verify that the inequality holds for $k\le 2$.

Now suppose the inductive assertion holds for $k = s \ge 2$, that is $f(s)\le 2^{2^s-ws+3}-6$.
Using Corollary~$\ref{cor:bound1}$ and the induction hypothesis we have
\begin{align*}
  f(s+1) &\le 2^{2^s} + 2^{2^{s}-w}f(s) \\
  &\le 2^{2^s} + 2^{2^s-w}(2^{2^s-ws+3}-6) \\
  &= 2^{2^s} + 2^{2^s-w+2^s-ws+3} - 2\cdot2^{2^{s}} \\
  &= 2^{2^{s+1}-w(s+1)+3} - 2^{2^{s}} \\
  &\le 2^{2^{s+1}-w(s+1)+3} - 6
\end{align*}
since $2^{2^{s}} \ge 6$ when $s\ge 2$. This concludes the induction proof.
\end{proof}

This means that the length of the common subsequence
$CS(k)$ is
\[2^{2^k} - f(k)\ge 2^{2^k} - 2^{2^{k}-wk+3}
= 2^{2^k}\left(1-\frac{1}{2^{wk}/8}\right)
= 2^{2^k}\left(1-\frac{1}{3^{k}/8}\right).
\]

This asymptotic behaviour propagate through the other generalizations,
and we obtain a slightly better versions of Corollaries~\ref{cor:an}
and \ref{cor:bn}.

\begin{thm}
  $a(n) = 2^{n}(1-\mathcal{O}(\frac{1}{n^w}))$
  and
  $b(n) = n\left(1-\mathcal{O}\left(\frac{1}{(\log n)^w}\right)\right)$
  where $w = \log_2(3) \approx 1.58$.
\end{thm}

\section{Acknowledgment}
I thank Jeffrey Shallit for telling me about the problem.

\bibliographystyle{elsarticle-num-names}

\begin{thebibliography}{1}

\bibitem[Allouche et al.(1999)]{allouche}
J.-P.~Allouche, J.~Shallit,
\newblock The ubiquitous Prouhet-Thue-Morse sequence.
\newblock \emph{Sequences and Their Applications: Proceedings of SETA '98}, Springer-Verlag, 1999, pp. 1-16

\bibitem[Jean Berstel(2006)]{berstel}
Jean~Berstel,
\newblock Combinatorics on Words Examples and Problems.
\newblock \url{http://www-igm.univ-mlv.fr/~berstel/Exposes/2006-05-24TurkuCow.pdf} (2006)

\bibitem{oeis}
N.~J.~A.~Sloane,
\newblock Online Encyclopedia of Integer Sequences.
\newblock \url{http://oeis.org}

\end{thebibliography}

\end{document}